\documentclass[12pt,titlepage]{article}
\usepackage[applemac]{inputenc}
\usepackage{anysize}
\usepackage{float}
\usepackage{amsmath}
\usepackage{amssymb}
\usepackage{amsbsy}
\usepackage{amsfonts}
\usepackage{amsthm}
\usepackage{latexsym}

\usepackage{amsmath,amssymb,amsthm}
\let\phi=\varphi

\newcommand{\Ker}{\operatorname{Ker}}
\newcommand{\wt}{\operatorname{wt}}
\newcommand{\id}{\operatorname{id}}
\newcommand{\Gal}{\operatorname{Gal}}
\newcommand{\ord}{\operatorname{ord}}
\newcommand{\tr}{\operatorname{tr}}
\newcommand{\rk}{\operatorname{rk}}
\newcommand{\F}{\mathbb{F}}
\newcommand{\apl}[3]{#1\colon #2\longrightarrow #3}
\newcommand{\Z}{\mathbb{Z}}

\newcommand{\GG}{\mathcal{G}}
\newcommand{\HH}{\mathcal{H}}
\newcommand{\PP}{\mathcal{P}}
\newcommand{\TT}{\mathbb{T}}

\newtheorem{lemma}{Lemma}

\newtheorem{theorem}[lemma]{Theorem}

\newtheorem{definition}{Definition}
\newtheorem{corollary}[lemma]{Corollary}
\newenvironment{example}{{\bf Example.\/}}

\begin{document}

~\bigskip
\bigskip
\bigskip
\bigskip

\thispagestyle{empty}

 \centerline{\Large The $GC$-content of a family of cyclic codes with applications to DNA-codes}
\bigskip
\bigskip
\bigskip
\centerline{by}
\bigskip
\bigskip
\bigskip

\centerline{\bf Josu Sangroniz}
\centerline{Departamento de Matem\'aticas}
\centerline{Facultad de Ciencia y Tecnolog\'\i a}
\centerline{Universidad del Pa\'{\i}s Vasco} \centerline{48080
Bilbao.} \centerline{SPAIN} \centerline{E-mail:
josu.sangroniz@ehu.es}

\smallskip
\centerline{and}
\smallskip

\centerline{\bf Luis Mart\'\i nez }
\centerline{Departamento de Matem\'aticas}
\centerline{Facultad de Ciencia y Tecnolog\'\i a}
\centerline{Universidad del Pa\'{\i}s Vasco} \centerline{48080
Bilbao.} \centerline{SPAIN} \centerline{E-mail:
luis.martinez@ehu.es}

\bigskip
\bigskip
\bigskip
\bigskip
\centerline{\bf Abstract}
\bigskip
 \centerline{\parbox{13cm}{
 Given a prime power $q$ and a positive integer $r>1$ we say that a cyclic code of length $n$, $C\subseteq\F_{q^r}^n$, is \emph{Galois supplemented} if for any
 non-trivial element $\sigma$ in the Galois group of the extension $\F_{q^r}/\F_q$, $C+C^\sigma=\F_{q^r}^n$, where $C^\sigma=\{(x_1^\sigma,\dots,x_n^\sigma)\mid
 (x_1,\dots,x_n)\in C\}$. This family includes the quadratic-residue (QR) codes over $\F_{q^2}$. Some important properties QR-codes are then    extended to
 Galois supplemented codes and a new one is also considered,  which is actually the motivation for the introduction of this family of codes: in a Galois
 supplemented code we can explicitly count the number of words that have a fixed number of coordinates  in $\F_q$.
In connection with DNA-codes the number of coordinates of a word in $\F_4^n$ that lie in $\F_2$ is sometimes referred to as  the \emph{$GC$-content} of the word
and codes over $\F_4$ all of whose words have the same $GC$-content have a particular interest. Therefore our results have some direct applications in this
direction.}}

\bigskip
\bigskip
\bigskip
\bigskip
\bigskip
\bigskip
\bigskip
\bigskip

\newpage
\section{Introduction}

In this work we consider a finite extension of finite fields $\F_{q^r}/\F_q$, $r>1$, and study codes over $\F_{q^r}$. A case to keep always in mind is the
quaternary extension $\F_4/\F_2$ because of its important applications  in genome technologies, a context in which these codes  are usually referred to as
DNA-codes. In principle DNA strands of a fixed length can be thought of as codewords over the alphabet $A$, $C$, $G$ and $T$ representing the nucleotide bases,
but we can immediately translate these symbols into the elements of $\F_4$.

DNA-codes satisfying certain specific properties are of particular interest, notably it is desirable that their minimum Hamming distance is large as well as the
distance between any codeword and the reverse-complement of any  word in the code. For our purposes  the \emph{reverse-complement} of a word
$v=(x_1,x_2,\dots,x_n)\in\F_4^n$ is $v^{RC}=(x_n+1,\dots,x_2+1,x_1+1)$. Another   restriction that is relevant in DNA technologies is that all codewords have
the same $GC$-content. Originally this means that the number of nucleotide bases $G$ and $C$, the \emph{$GC$-content}, in all   codewords is the same although
we  reinterpret this condition by requiring that all codewords have the same number of coordinates in the base field $\F_2$ (of course we can do this by simply
identifying the nucleotides $G$ and $C$ with the elements in $\F_2$). If we denote by $d(u,v)$ the Hamming distance between  $u,v\in\F_4^n$, the maximum number
of words in a DNA-code $C\subseteq\F_4^n$ satisfying the three conditions: 1) $d(u,v)\geq d$ for all $u,v\in C$, $u\not=v$; 2) $d(u,v^{RC})\geq d$ for all
$u,v\in C$ and 3) all codewords have the same $GC$-content $w$, is denoted $A_4^{GC,RC}(n,d,w)$. Since the concrete value of $w$ is not so relevant, we consider
$A_4^{GC,RC}(n,d)=\max_w A_4^{GC,RC}(n,d,w)$. Lower bounds for this number and  some particular values of the parameters $n$ and $d$ can be found in
\cite{Aboluion}.

We shall show that for an easily describable family of cyclic codes over $\F_{q^r}$ it is possible to count the number of codewords with a given number of
coordinates  in $\F_q$. Instead of using the term $GC$-content, we'll prefer to call this number the \emph{$\F_q$-weight} of the word. As an immediate
consequence a general lower bound for $A_4^{GC,RC}(n,d)$ will follow, which in some cases  improves the bounds in \cite{Aboluion}. These codes include the
QR-codes (quadratic-residue codes) over $\F_{q^2}$.

The paper is organized as follows. In Section 2 we give a general formula to count the number of codewords of a cyclic code over $\F_{q^r}$ with a given
$\F_q$-weight $w$. This formula  depends only on the parameters $q$, $r$, $n$ and $w$ for a family of cyclic codes that is defined in Section 3, and that we
call \emph{Galois supplemented} codes. We can also find the corresponding formulas for the even weight subcodes and the extended codes of these cyclic codes. In
Section 4 we show how to construct Galois supplemented codes using cyclotomic cosets and it will become clear then that QR-codes over $\F_{q^2}$  are Galois
supplemented.  There is also some overlapping between quaternary Galois supplemented codes and the family of $Q$-codes defined by V.~Pless. These codes were
defined in \cite{Pless} in terms of their idempotent polynomials. The relation between the two families is clarified in Section 5, where we compute the
idempotent polynomials for some quaternary Galois supplemented codes. We also study in this section other general properties  of these codes (minimal odd
weight, automorphisms and duality). Finally, in the last section, some applications to DNA-codes are presented.

\section{Enumerating words with a fixed $\F_q$-weight}\label{enumerating}

\begin{definition} Let $u\in\F_{q^r}^n$. We call the \emph{$\F_q$-weight} of $u$, and we denote it $\wt_{\F_q}(u)$, to the number of coordinates of $u$ lying in
$\F_q$. If $C\subseteq\F_{q^r}^n$ is a code over $\F_{q^r}$ we define the  \emph{$\F_q$-weight enumerator polynomial} of $C$, $W_{\F_q,C}$, as
$$W_{\F_q,C}(X,Y)=\sum_{u\in C}X^{\wt_{\F_q}(u)}Y^{n-\wt_{\F_q}(u)}=\sum_{w\geq 0}b_wX^wY^{n-w},$$
where $b_w=b_w(C)$ is the number of codewords in $C$ of $\F_q$-weight $w$.
\end{definition}

Let $\sigma$ be the Frobenius automorphism of the extension   $\F_{q^r}/\F_q$ and $\psi$ the linear map (in the sequel  \emph{linear} means $\F_q$-linear)
$\psi=\sigma-\id$. These maps extend naturally to the polynomial algebra $\F_{q^r}[x]$ by acting on the coefficients. For $g\in\F_{q^r}[x]$, the image by $\psi$
is denoted as usual $\psi(g)$ but we we prefer the notation $g^\sigma$ for the image by $\sigma$. In this way $\sigma$ becomes an algebra automorphism and
$\psi$ a lineal map. Notice that the kernel of $\psi$ is $\F_q[x]$ and its image consists of the polynomials in  $\F_{q^r}[x]$ whose coefficients have zero
trace.

The maps $\sigma$ and $\psi$ fix setwise the ideal $(x^n-1)\subseteq \F_{q^r}[x]$, so they induce well-defined maps from the quotient algebra
$\F_{q^r}[x]/(x^n-1)$ to itself, that we still denote the same. Now the kernel of $\psi$ is the image of    $\F_q[x]$ in $\F_{q^r}[x]/(x^n-1)$ and it is a
subspace of dimension $n$.

If $f\in\F_{q^r}[x]$, the number of zero coefficients of  $\psi(f)$ is the same as the number of coefficients of $f$ in $\F_q$, so under the natural
identification of $\F_{q^r}^n$  and $\F_{q^r}[x]/(x^n-1)$, the $\F_q$-weight of $\overline f$ is $n-\wt (\psi(\overline f))$ ($\overline f$ is the image of $f$
in $\F_{q^r}[x]/(x^n-1)$ and  $\wt(u)$  the usual weight of   $u\in\F_{q^r}^n$).

For the rest of the paper   $n$ will be coprime with $q$ and $g\in\F_{q^r}[x]$ will be a divisor of $x^n-1\in\F_q[x]$. We'll denote by $C$ (or $C_g$ or even
$C_g^{\F_{q^r}}$) the cyclic code over $\F_{q^r}$ with generator polynomial $g$, that is $$C=(g)/(x^n-1)=\{\overline f\in \F_{q^r}
[x]/(x^n-1) \text{ such that } g\mid f\}.$$

Let's consider now $\psi_{|C}$. Its kernel is of course $C\cap\Ker\psi=\{\overline f\mid f\in\F_q[x],\, g\mid f\}$. If $\overline f$ lies in this set, we
certainly have that $g^{\sigma^i}\mid f$ for all $0\leq i<r$, so $[g,g^\sigma,\dots ,g^{\sigma^{r-1}}]\mid f$ (as usual,  brackets mean least common multiple
and parentheses greatest common divisor), so the kernel of  $\psi_{|C}$ consists of the elements $\overline f$ such that $f\in\F_q[x]$ is a multiple of
$[g,g^\sigma,\dots ,g^{\sigma^{r-1}}]$, a polynomial that lies in   $\F_q[x]$. Therefore $\Ker \psi_{|C}$ is naturally isomorphic (as  $\F_q$-vector space) to
$(\overline{[g,g^\sigma,\dots ,g^{\sigma^{r-1}}]})\subseteq \F_q[x]/(x^n-1)$, whose dimension is $n-\deg [g,g^\sigma,\dots , g^{\sigma^{r-1}}]$. The next result
follows now almost directly.

\begin{theorem}\label{main th} Let $C\subseteq \F_{q^r}^n$  be the cyclic code with generator polynomial $g$. Then $$W_{\F_q,C}=q^{n-\deg [g,g^\sigma,\dots ,
g^{\sigma^{r-1}}]}W_{\psi(C)}^R,$$
where $W_{\psi(C)}^R$ is the reciprocal polynomial of the weight enumerator polynomial of  $\psi(C)$ (the \emph{reciprocal polynomial} of a bivariate polynomial
$W(X,Y)$ is $W^R(X,Y)=W(Y,X)$; for a univariate polynomial $f$ of degree $n$, it is $f^R(x)=x^nf(1/x)$).
\end{theorem}
\begin{proof}
The elements in $C$ with $\F_q$-weight $w$ are exactly those whose image by $\psi$ has weight $n-w$, so $$b_w=|\Ker\psi_{|C}|a_{n-w}=q^{n-\deg [g,g^\sigma,\dots
, g^{\sigma^{r-1}}]}a_{n-w},$$
where $a_w=a_w(\psi(C))$ is the number of elements in  $\psi(C)$ of weight $w$. The result follows immediately.
\end{proof}

There are   two cases in which the code  $\psi(C)$ has a particularly simple description. The first is when $r=2$ and we deal with it in the next result. The
other case leads to the introduction of a new family of codes that is  to be  analyzed in the rest of the paper.

\begin{theorem} Let $r=2$, $\lambda\in\F_{q^2}$, $\lambda\not=0$, a fixed non-zero element of zero trace. Then $\psi(C)=\lambda J$, where $J$ is the image of
the cyclic code $C_{(g,g^\sigma)}^{\F_q}$ in $\F_{q^r}[x]/(x^n-1)$. Therefore
$$W_{\F_q,C}=q^{n-\deg [g,g^\sigma]}W_{C_{(g,g^\sigma)}^{\F_q}}^R.$$
\end{theorem}
\begin{proof}

As $r=2$,
$$\dim \psi(C)=\dim C-\dim\Ker\psi_{|C}=2(n-\deg g)-(n-\deg [g,g^\sigma])=n-\deg (g,g^\sigma).$$

On the other hand, the zero-trace elements in $\F_{q^2}$ form a $1$-dimensional $\F_q$-subspace, so any element with zero trace is one of $\alpha\lambda$ for
some $\alpha\in\F_q$ and any polynomial in the image of $\psi$ can be written as $\lambda f_1$ with $f_1\in\F_q[x]$. It is also clear that for any $\overline f$
in $\psi(C)$, the polynomial $(g,g^\sigma)$ divides $f$. But $(g,g^\sigma)$ has coefficients in $\F_q$ (given that
$(g,g^\sigma)^\sigma=(g^\sigma,g^{\sigma^2})=(g^\sigma,g)$), so
 $$\psi(C)\subseteq \{\lambda\overline f_1\mid f_1\in ((g,g^\sigma))\subseteq \F_q[x]\}=\lambda J.$$
Now, the inclusion of  $\F_q[x]$ in $\F_{q^r}[x]$ induces an injective map between the quotients by the corresponding ideals generated by  $x^n-1$, so $J$ has
the same   dimension   as the   cyclic code over  $\F_q$ generated by $(g,g^\sigma)$, namely $n-\deg (g,g^\sigma)$. We conclude that both $\psi(C)$ and $\lambda
J$ have the same dimension, whence they actually coincide.
\end{proof}

\section{$\F_q$-weights in Galois supplemented   and related codes}

\begin{definition} Let $C\subseteq\F_{q^r}^n$ be a cyclic code of length $n$ over $\F_{q^r}$. We say that $C$ is \emph{Galois supplemented} if
$C+C^\sigma=\F_{q^r}^n$ for any non-trivial automorphism $\sigma\in\Gal(\F_{q^r}/\F_q)$, the Galois group of the extension $\F_{q^r}/\F_q$,  where
$C^\sigma=\{(x_1^\sigma,\dots,x_n^\sigma)\mid (x_1,\dots,x_n)\in C\}$.
\end{definition}

If $C\subseteq\F_{q^r}^n$ is a cyclic code with generator polynomial $g\in\F_{q^r}[x]$ and $\sigma\in\Gal(\F_{q^r}/\F_q)$, it is clear that $C^\sigma$ is the
cyclic code with generator polynomial $g^\sigma$, so the two conditions $\F_{q^r}^n=C+C^\sigma$ and $(g,g^\sigma)=1$ are equivalent. This motivates the
following definition.

\begin{definition}
Let $g\in\F_{q^r}[x]$. We say that $g$ is \emph{Galois coprime} if $(g,g^\sigma)=1$ for any non-trivial $\sigma\in\Gal (\F_{q^r}/\F_q)$.
\end{definition}

It is clear now that a cyclic code $C\subseteq\F_{q^r}^n$ is Galois supplemented if and only if its generator polynomial is Galois coprime. In the sequel $C$
will always denote a cyclic code with generator polynomial $g$. Notice that if $C$ is Galois supplemented,  the polynomials $g,g^\sigma,\dots,g^{\sigma^{r-1}}$
are pairwise coprime (now $\sigma$ denotes the Frobenius automorphism of $\F_{q^r}/\F_q$), so $[g,g^\sigma,\dots,g^{\sigma^{r-1}}]=gg^\sigma\dots
g^{\sigma^{r-1}}$, whence
$$\dim \psi(C)=r(n-\deg g)-(n-r\deg g)=(r-1)n.$$
But $\psi(C)$ is clearly contained in the space formed by  the elements $\overline f\in \F_{q^r}[x]/(x^n-1)$, where $f$ is a polynomial with degree less than
$n$ and zero-trace coefficients. The dimension of this space is precisely $(r-1)n$, so it coincides with $\psi(C)$. Now it is clear that the number of elements
in $\psi(C)$ of weight $w$, $0\leq w\leq n$, is  ${n\choose w}(q^{r-1}-1)^w$ (because in $\F_{q^r}$ there are exactly $q^{r-1}-1$ non-zero elements with zero
trace). The next result follows immediately from Theorem \ref{main th}.

\begin{theorem}\label{principal} Let $C\subseteq\F_{q^r}^n$ be a Galois supplemented   code with generator polynomial $g$. Then
$$W_{\F_q,C}(X,Y)=q^{n-r\deg g}(X+(q^{r-1}-1)Y)^n.$$
\end{theorem}

Alternatively, Theorem \ref{principal} also follows easily from another    characterization of Galois supplemented codes: they are  the cyclic codes $C\subseteq
\F_{q^r}^n$ satisfying $C+\F_q^n=\F_{q^r}^n$. (Indeed, if $g$ is the generator of $C$, $C\cap \F_q^n$ is the cyclic code over $\F_q$ with generator polynomial
$[g,g^\sigma,\dots,g^{\sigma^{r-1}}]$. By comparing the dimensions of $C+\F_q^n$ and $\F_{q^r}^n$ one gets that they are equal if and only if $\deg
[g,g^\sigma,\dots,g^{\sigma^{r-1}}]=r\deg g$, that is $g$ is Galois coprime.)

Another interesting approach that we should mention is using MacWilliams theorem for the complete weight enumerator  polynomial (see \cite[Chapter 5, Theorem
10]{MacWilliamsSloane}), as is done in \cite[Theorem 1]{GaboritKing} for the extension $\F_4/\F_2$. To state the corresponding $q^r$-ary version we need to
introduce one more element. In $\F_{q^r}^n$ we have the usual dot product $\cdot$ and we can define  a new one $\circ$ by setting $u\circ v=f(u\cdot v)$, where
$\apl f {\F_{q^r}}{\F_p}$ is any non-trivial linear form on the $\F_p$-vector space $\F_{q^r}$ ($p$ is of course the characteristic of $\F_{q^r}$). Now, $\circ$
is a non-degenerate $\F_p$-bilinear form on $\F_{q^r}^n$. We'll use $\perp_\circ$ to denote orthogonality with respect to this form and $\perp$ for the usual
orthogonality with respect to $\cdot$. In fact one has $W^\perp=W^{\perp_\circ}$ for any $\F_{q^r}$-subspace $W\subseteq\F_{q^r}^n$. MacWilliams Theorem yields
then that for any linear code $C\subseteq\F_{q^r}^n$,
$$W_{\F_q,C^\perp}(X,Y)=\frac{q^n}{|C|}W_{C\cap (\F_q^n)^{\perp_\circ}}(X+(q^{r-1}-1)Y,\,X-Y).$$
If $g(x)h(x)=x^n-1$, when we apply this formula for the cyclic code $C=C_{h^R}$ we get
\begin{equation}\label{MacWilliams}
W_{\F_q,C_g}(X,Y)=q^{n-r\deg g}W_{(C_g + \F_q^n)^{\perp_\circ}}(X+(q^{r-1}-1)Y,\,X-Y).\end{equation}
(Notice that $C_{h^R}=C_g^\perp=C_g^{\perp_\circ}$.) Now Theorem \ref{principal} is clear since, for a Galois coprime polynomial $g$,
$(C_g+\F_q^n)^{\perp_\circ}=(\F_{q^r}^n)^{\perp_\circ}=\{0\}$. In fact (\ref{MacWilliams}) shows that $g$ is Galois coprime if and only if the $\F_q$-weight
enumerator polynomial of $C_g$ is as in Theorem \ref{principal}.

Next we would like to compute the $\F_q$-weight enumerator polynomial for
 the even weight subcode and for the extended code of a Galois supplemented code.  Recall that  the even weight subcode $C^+$ of a code $C$ is
 $C^+=\{(x_1,\dots,x_n)\in C\mid x_1+\cdots +x_n=0\}$. If $C$ is a cyclic code with generator polynomial $g$, this is a proper subcode if and only if
 $g(1)\not=0$ and in this case it is in fact the cyclic code with generator polynomial $g_1(x)=(x-1)g(x)$. So we can consider the case when $C$ is a Galois
 supplemented code with generator polynomial $g$ or, even more generally, the cyclic subcode of $C$ with generator polynomial $g_0g$,  $C_{g_{\raise -.5pt
 \hbox{$\scriptscriptstyle 0$}}g}$, where $g_0\in\F_q[x]$ is any divisor of $x^n-1$. With the notation in Section \ref{enumerating} it is clear that
\begin{equation}\label{pluscode} \psi(C_{g_{\raise -.5pt  \hbox{$\scriptscriptstyle 0$}}g}) \subseteq \{\overline f\in \F_{q^r}[x]/(x^n-1)\mid
\tr_{\F_{q^r}/\F_q}f=0\text{ and } g_0\mid f\},\end{equation}
where, $\tr_{\F_{q^r}/\F_q}f=0$   means that all the coefficients of $f$ have zero trace. The dimension of the subspace in the right-hand side of
(\ref{pluscode}) is   $(r-1)(n-\deg g_0)$ and the dimension of $\psi(C_{g_{\raise -.5pt  \hbox{$\scriptscriptstyle 0$}}g})$ is
\begin{eqnarray*}
\dim C_{g_{\raise -.5pt  \hbox{$\scriptscriptstyle 0$}}g} -\dim \Ker\psi_{|C_{{g_{\raise -.5pt  \hbox{$\scriptscriptstyle 0$}}}g}}&=&r(n-\deg g-\deg
g_0)-(n-\deg [g_0g,\dots ,g_0g^{\sigma^{r-1}}])\\&=&
r(n-\deg g-\deg g_0)-(n-r\deg g-\deg g_0)=(r-1)(n-\deg g_0),\end{eqnarray*}
so the inclusion (\ref{pluscode}) is in fact an equality.

Next we want to count how many elements in $\psi(C_{g_{\raise -.5pt  \hbox{$\scriptscriptstyle 0$}}g})$ have weight $N=n-w$. Since the polynomial $g_0$ has
coefficients in $\F_q$, the condition that $g_0\mid f$ can be reformulated in terms of the coefficients of $f$ being annihilated by a certain matrix over
$\F_q$, thus we want to count the number of words of weight $N$ in a set of the form
$$\{(x_1,\dots ,x_n)\in\F_{q^r}^n\mid \tr_{\F_{q^r}/\F_q}(x_i)=0,\ i=1,\dots ,n,\ (x_1,\dots ,x_n)H=0 \},$$
where $H$ is a matrix over $\F_q$. More generally, suppose $\TT$ is a finitely generated $\F_q$-vector space contained in any field extension $\F$ of $\F_q$ and
consider the set
$\{x\in\TT ^n \mid  xH=0 \}$. The nullspace $W\subseteq\F^n$ of $H$ has an $\F$-basis $v_1,\dots ,v_s$, $s=n-\rk H$, with $v_i\in\F_q^n$ and it is clear that
$\lambda_1v_1+\cdots +\lambda_sv_s\in\TT^n$ if and only if $\lambda_i\in\TT$ for all $1\leq i\leq s$, so the number of elements in $\TT^n$ annihilated by $H$ is
$|\TT|^s$, which, for a fixed $H$, only depends on the size of $\TT$. Now we can use the inclusion-exclusion principle to count how many of these elements have
weight $N$. If we fix the set $S\subseteq\{1,\dots,n\}$ of the positions of the non-zero coordinates, this number is
$$\sum_{T\subseteq S} (-1)^{|T|}|\TT|^{N-|T|-\rk H_{S\backslash T}}=\sum_{U\subseteq S}
(-1)^{N-|U|}|\TT|^{|U|-\rk H_U},$$
where $H_U$ is the submatrix of $H$ consisting of the rows indexed by the elements in $U$. Again this number only depends on the size of $\TT$. Therefore, the
weight enumerator polynomial of $\{x\in\TT ^n \mid  xH=0 \}$ is
\begin{equation}\label{polynomial}
\sum_{T\cap U=\emptyset}
(-1)^{|T|}|\TT|^{|U|-\rk H_U} X^{|T|+|U|}Y^{n-|T|-|U|}=
\sum_{U\subseteq\{1,\dots ,n\}} |\TT|^{|U|-\rk H_U}X^{|U|}(Y-X)^{n-|U|}.\end{equation}
This will be in particular the weight enumerator polynomial of the code over $\F_{q^d}$ with control matrix $H$, where $d$ is the dimension of $\TT$. In our
case the set $\TT$ consists of the elements in $\F_{q^r}$ with zero trace, which has dimension $r-1$, so the following theorem is now a consequence of the
preceding discussion and Theorem \ref{main th}.

\begin{theorem} Suppose that $g\in \F_{q^r}[x]$ is Galois coprime, $g_0\in\F_q[x]$ and $g_0g\mid x^n-1$. Then
$$W_{\F_q,C_{g_{\raise -.5pt  \hbox{$\scriptscriptstyle 0$}}g}}=q^{n-r\deg g-\deg g_0}(W_{{C_{g_{\raise -.5pt  \hbox{$\scriptscriptstyle 0$}}}}^{\kern
-4pt{\F_{\kern -1pt  q^{r-1}}}}})^R.$$
\end{theorem}

Of course, when we set $g_0=1$, this theorem is consistent with Theorem \ref{principal} as $W_{\F_{q^{r-1}}^n}(X,Y)=((q^{r-1}-1)X+Y)^n$. The weight enumerator
polynomial for the zero-parity  code is also known (see \cite[Example 1.45]{Xambo} or, alternatively, compute it by using (\ref{polynomial})), so we have:

\begin{corollary}\label{enumeratorplus} Let $C\subseteq \F_{q^r}^n$ be a Galois supplemented code with generator polynomial $g$ and $C^+$ the even weight
subcode. Then
$$W_{\F_q,C^+}(X,Y)=q^{n-r(\deg g+1)}((X+(q^{r-1}-1)Y)^n+(q^{r-1}-1)(X-Y)^n).$$
\end{corollary}

The following construction can be regarded as a kind of   dual to the passing from $C_g$ to $C_{g_{\raise -.5pt  \hbox{$\scriptscriptstyle 0$}}g}$. We start
with a polynomial $g_0\in\F_q[x]$ dividing $x^m-1$, where $m>n$ and $\deg g_0=m-n$. Then for any $f\in\F_{q^r}[x]$ of degree   $<n$ there exists a unique
polynomial $r\in\F_{q^r}[x]$ of degree $<m-n$ such that $g_0$ divides $f(x)+r(x)x^n$. This gives an embedding $C_g\hookrightarrow C_{g_0}$ whose image
$\widetilde C_g$ will be called the \emph{extended} code of $C_g$ by $g_0$. Of course, when $g_0(x)=x-1$ this construction gives the usual extended code
obtained by adding the parity check coordinate. We have the same inclusion as in (\ref{pluscode}) with $C_{g_{\raise -.5pt  \hbox{$\scriptscriptstyle 0$}}g}$
and $n$ replaced by $\widetilde C_g$ and $m$, respectively, which is in fact an equality, being the dimensions of the two spaces the same, namely
$(r-1)n=(m-\deg g_0)(r-1)$. Now we can repeat verbatim the same argument as before to conclude the following result.

\begin{theorem} Suppose that $g\in \F_{q^r}[x]$ is Galois coprime and $g_0\in\F_q[x]$ with $g\mid x^n-1$, $g_0\mid x^m-1$, $\deg g_0=m-n$. Then
$$W_{\F_q,\widetilde C_g}=q^{m-r\deg g-\deg g_0}(W_{{C_{g_{\raise -.5pt  \hbox{$\scriptscriptstyle 0$}}}}^{\kern -4pt{\F_{\kern -1pt  q^{r-1}}}}})^R.$$
\end{theorem}

\begin{corollary} Let $C\subseteq \F_{q^r}^n$  be a Galois supplemented code with generator polynomial $g$ and $\widetilde C$ the (usual) extended code. Then
$$W_{\F_q,\widetilde C}(X,Y)=q^{n+1-r(\deg g+1)}((X+(q^{r-1}-1)Y)^{n+1}+(q^{r-1}-1)(X-Y)^{n+1}).$$
\end{corollary}

\section{Constructing Galois supplemented codes}

Now we want to characterize Galois supplemented  codes in terms of  cyclotomic cosets. We denote by $\F$  the splitting field of $x^n-1$ over $\F_{q^r}$ and fix
a primitive $n$th root of unity $\zeta\in\F$. The (monic) divisors of $x^n-1$ in $\F[x]$ correspond bijectively with the subsets $B\subseteq\Z/n\Z$ via the map
$B\mapsto g_B$, where $g_B(x)=\prod_{\overline k\in B}(x-\zeta^k)$ (of course, it makes sense to write $\zeta^k$ for $\overline k\in\Z/n\Z$). For ease of
notation we will drop the bar when writing elements of $\Z/n\Z$ but it should be noted that operations between integers must be then thought to be done modulo
$n$. The polynomial $g_B$  has then coefficients in $\F_{q^r}$ if and only if $q^rB=B$, so if  $C_B$  is the cyclic code with generator polynomial $g_B$, then
the map $B\mapsto C_B$ is a bijection between the subsets $B\subseteq\Z/n\Z$ satisfying $q^rB=B$ and the cyclic codes over $\F_{q^r}$ of length $n$. Our goal is
to pinpoint the $B$'s for which $C_B$ is Galois supplemented.

We shall need two elementary lemmas from group theory. Recall that if a group $G$ acts on a set $\Omega$ (we write actions on the left) a non-empty subset
$B\subset\Omega$ is called a \emph{block} for the action if for all $g\in G$, either $gB=B$ or else  $B\cap gB=\emptyset$. Notice that against the usual
convention when speaking about blocks we do not require the action to be transitive. The \emph{stabilizer} of the block $B$ is the subgroup of $G$ formed by the
elements $g\in G$ such that $gB=B$. We'll denote it $G_B$. Clearly, a subgroup  $H\leq G$ is contained in the stabilizer of $B$ if and only if $B$ is a union of
$H$-orbits.

\begin{lemma}\label{bloque} Suppose a group $G$ acts on a set $\Omega$ and $H$ is a subgroup of $G$. Then $B\subset \Omega$ is a block for the action with
stabilizer $H$ if and only if $B$ is a union of $H$-orbits, not two of which are in the same  $G$-orbit and for all
 $b\in B$  the stabilizer of $b$ in $G$, $G_b$, is contained in $H$.
\end{lemma}

\begin{proof}
If $B$ is a block with stabilizer $H$,  $B$ is a union of $H$-orbits as we have just said  and, moreover, if two of them, say  $Hb$ and $Hb'$, are contained in
the same  $G$-orbit then $b'=gb$ for some $g\in G$ and $b'\in gB\cap B$, whence $gB=B$, that is $g\in G_B=H$, so $Hb=Hb'$. In addition, if $g\in G_b$, $b=gb\in
gB\cap B$, so $gB=B$ and again $g\in G_B=H$.

Conversely, assume  $B\subset \Omega$ satisfies the conditions stated in the lemma. It is clear that $B$ is invariant by $H$. If for some $g\in G$, $gB\cap B$
is non-empty, then there exist $b,b'\in B$ such that $b'=gb$ and $Hb'$ and $Hb$ are contained in the same $G$-orbit. Hence $Hb=Hb'$ and in fact $b'=hb$ for some
$h\in H$. Then  $h^{-1}g\in G_b\subseteq H$ , which implies $g\in H$ and $B=gB$. So $B$ is a block with stabilizer $H$.
\end{proof}

\begin{lemma}\label{epim} Let $\GG$ and $G$ be two groups acting on sets $R$ and $\Omega$, respectively. Suppose $\apl\pi \GG G$ is a group epimorphism and
$\apl\beta R\Omega$ is a bijection such that $\beta (\sigma\zeta)=\pi(\sigma)\beta(\zeta)$ for all     $\sigma\in\GG$ and $\zeta\in R$. Let $\HH$ be a subgroup
of  $\GG$ and  $H=\pi(\HH)$. Then $R_0\subset R$ is a block for $\GG$ with stabilizer $\HH$ if and only if   $\Ker\pi\subseteq \HH$ and $B=\beta (R_0)$ is a
block for $G$ with stabilizer $H$.
\end{lemma}

\begin{proof} The hypotheses on $\pi$ and $\beta$ imply that $\Ker\pi$ is contained in the kernel of the action of $\GG$ on $R$, so we can consider
$\GG/\Ker\pi$ acting on $R$ and then, via the isomorphism between $\GG/\Ker\pi$ and $G$ and the bijection $\beta$, this action corresponds to the action of $G$
on $\Omega$. It is clear that, under these identifications, action blocks and stabilizers correspond.
\end{proof}

The group of units of $\Z/n\Z$, $(\Z/n\Z)^\times$, acts by multiplication on $\Omega=\Z/n\Z$. Since $(q,n)=1$,  $G=\langle q\rangle\leq (\Z/n\Z)^\times$ acts on
$\Omega$. The orbits of the action of $G$ on $\Omega$ are sometimes called in the literature \emph{$q$-cyclotomic cosets}, but we prefer to refer to them simply
as $G$-orbits. Let $H=\langle  q^r\rangle\subseteq G$.
Sometimes we will have to consider $q$ as an element of  $(\Z/m\Z)^\times$ for some divisor  $m$ of $n$. We'll write $\ord (q \mod m)$ for its order in this
group. Note that if $m\mid m'$, $\ord (q\mod m)\mid \ord (q\mod m')$.

\begin{theorem}\label{principal2} Let $B\subseteq\Z/n\Z$. Then $g_B\in\F_{q^r}[x]$ is Galois coprime   if and only if  $B$ is a union of $H$-orbits, not two of
which are in the same $G$-orbit and $r\mid \ord (q\mod m_k)$, $m_k=n/(n,k)$, for all $k\in B$.
\end{theorem}

\begin{proof}
The group $\GG=\Gal(\F/\F_q)$ (recall that $\F$ is the splitting field of $x^n-1$ over $\F_{q^r}$) acts naturally on the set $R$ of the $n$th roots of unity
contained in $\F$. We fix a primitive one, $\zeta\in R$. Now $R$ can be identified with  $\Omega=\Z/n\Z$ via the map $\apl\beta R\Omega$ given by $\beta
(\zeta^k)=k$. On the other hand, there is a natural group epimorphism $\pi$ from  $\GG$ onto $G$ given by $\sigma\mapsto  q$, where $\sigma\in\GG$ is the
Frobenius automorphism. This way the action of   $\GG$ on $R$ corresponds with the action of $G$ on $\Omega$ in the sense of Lemma \ref{epim}. The statement
that $g_B\in\F_{q^r}[x]$ is Galois coprime is then equivalent  to $R_0=\{\zeta^k\mid   k\in B\}\subseteq R$ being a block for  $\GG$ with stabilizer
$\HH=\Gal(\F/\F_{q^r})=\langle\sigma^r\rangle$. The image of $\HH$ in $G$ is of course  $H$  so, by Lemma \ref{epim},  $g_B\in\F_{q^r}[x]$ is Galois coprime if
and only if $\Ker\pi\subseteq\HH$ and $B$ is a block for $G$ with stabilizer $H$.

Set $s=\ord(q\mod n)$. Then $\Ker\pi=\langle \sigma^s\rangle$. Since both $r$ and $s$ are divisors of the order of  $\sigma$ (this order is actually the least
common multiple of $r$ and $s$) we conclude that $\Ker\pi\subseteq\HH$ if and only if $r\mid s$. On the other hand the fact that $B$ is a block for  $G$ with
stabilizer $H$ is equivalent by Lemma \ref{bloque} to $B$ being a union of $H$-orbits, not two of which are in the same $G$-orbit and $G_k\subseteq H$ for all
$k\in B$ ($G_k$ is the stabilizer in $G$ of $k$).

We only have to show that if  $k\in B$, $G_k\subseteq H$ if and only if $r\mid \ord(q\mod m_k)=s_k$. Indeed, $q^j\in G$ fixes $k\in B$ if and only if
$q^jk\equiv k\pmod n$, that is $q^j\equiv 1 \pmod {m_k}$  or equivalently $s_k\mid j$. Therefore $G_k=\langle  q^{s_k}\rangle$. Given that $r$ and $s_k$ are
divisors of $s=\ord(q\mod n)$, we conclude that $G_k\subseteq H$ if and only if $r\mid s_k$.
\end{proof}

For an easy reference, we will say that the sets $B$ that satisfy the conditions in the last lemma are
 \emph{$q^r$-blocks} in $\Z/n\Z$. So, for fixed $q$, $r$ and $n$, the correspondences $B\mapsto g_B$ and $B\mapsto C_B$ establish  bijections between the
 $q^r$-blocks of  $\Z/n\Z$ and the  Galois coprime divisors of $x^n-1$ in $\F_{q^r}[x]$ and  the Galois supplemented   codes over $\F_{q^r}$ of length $n$,
 respectively.

It follows from Theorem \ref{principal2} that Galois supplemented codes over $\F_{q^r}$ and length $n$ exist if and only if $r\mid\ord(q\mod n)$. In the
important case $\F_4/\F_2$ this is equivalent to $2$ having even order modulo some prime divisor $p$ of $n$. By a classical result of Hasse's \cite{Hasse} this
is known to happen for an (asymptotic) average of $17$ primes out of every $24$.

\begin{example}
Let $q$, $r$ and $n$ be as usual and suppose $r\mid\ord (q\mod n)$. Suppose   $B\leq (\Z/n\Z)^\times$ is a subgroup such that $B\cap G=H$, where $G=\langle
q\rangle$ and $H=\langle  q^r\rangle\leq (\Z/n\Z)^\times$. Under these conditions, the hypotheses in the last theorem are immediate to check, so
$g_B\in\F_{q^r}[x]$ is Galois coprime. This happens in particular if $r=2$, $q$ is a quadratic non-residue modulo $n$ and $B$ is the set of quadratic residues
modulo $n$, so quadratic-residue codes over $\F_{q^2}$ are Galois supplemented and  Theorem \ref{principal} is valid for them. We don't even need to take $B$ to
be the whole set of quadratic-residues, any subgroup of it containing $q^2$ would do.
\end{example}

\section{Some properties of Galois supplemented codes}

Evidently, if a polynomial $g\in\F_{q^r}[x]$ is Galois coprime, $g(1)\not=0$, so the repetition code is always contained in a Galois supplemented code.

\begin{definition} We say that a Galois supplemented code $C\subseteq\F_{q^r}^n$ is \emph{complete} if the intersection of all the codes  $C^\sigma$,
$\sigma\in\Gal(\F_{q^r}/\F_q)$, is exactly the repetition code. Equivalently, when its generator polynomial $g=g_B$ satisfies one of the following equivalent
conditions:
\begin{enumerate}
\item $\prod_{\sigma\in\Gal(\F_{q^r}/\F_q)}g^\sigma(x)=1+x+\cdots +x^{n-1}.$
\item $\cup_{k=0}^{r-1}q^kB=\Z/n\Z\backslash\{0\}$ (disjoint union).
\end{enumerate}
In this case we will also say that $g$ and $B$ are \emph{complete}.
\end{definition}

It follows from Theorem \ref{principal2} that, given $q$, $r$ and $n$, Galois supplemented complete codes exist if and only if for any prime divisor $p$ of $n$,
the order of $q$ modulo $p$  is a multiple of $r$.

\begin{example}
Suppose an irreducible, Galois coprime complete polynomial $g\in\F_{q^r}[x]$ exists such that $g(x)\mid x^n-1$. Then the irreducible polynomial over $\F_{q^r}$
of an $n$th primitive root of unity  $\zeta$ is one of the polynomials $g^\sigma$, $\sigma\in\Gal (\F_{q^r}/\F_q)$, whence
$$\frac{n-1}{r}=\deg g =\deg g^\sigma =|\F_{q^r}(\zeta):\F_{q^r}|=\frac{\ord(q \mod n)}{r},$$
(remember that $r\mid\ord(q \mod n)$ by Theorem \ref{principal2}) so $\ord(q \mod n)=n-1$ and this implies that $n$ is prime and $q$ is a primitive root modulo
$n$.

Conversely, if   $n$ is prime, $r\mid n-1$ and $q$ is a primitive root modulo $n$, the polynomial $1+x+\cdots +x^{n-1}$ factorizes in $\F_{q^r}[x]$ as the
product of $r$ irreducible polynomials of degree $(n-1)/r$. If $H=\langle q^r\rangle$, these polynomials are $g_{q^kH}$, $k=0,\dots ,r-1$, so they are Galois
coprime and complete. Of course in this situation $H$ consists of the $r$th powers of the non-zero residues modulo $n$, so the corresponding code $C_H$ is the
code of the $r$th powers.
\end{example}

\subsection{Minimum weight of words of odd type}

\begin{lemma}\label{ceros} Let $K$ be a field and   $h(x)\in K[x]$. Suppose that among the coefficients of $h$ corresponding to the powers of $x$ of degree $<
n$, $k$ of them are the same and not zero. Then the polynomial  $h(x)(x^n-1)$ has at least $k$ terms of degree $\geq n$ with non-zero coefficient.
\end{lemma}

\begin{proof}
Let $\lambda\not=0$ be the common coefficient for the $k$ terms of $h$ as indicated in the statement and write  $h(x)=h_0(x)+h_1(x)x^n+\cdots +h_s(x)x^{sn}$,
$\deg h_i<n$. Let $X_0$ be the set of the indices $0\leq j<n$
such that the coefficient of $x^j$ is $\lambda$. For $i\geq 1$, we define inductively $X_i$ as the set of the $j\in X_{i-1}$ such that the coefficient of $x^j$
in $h_i(x)$ is $\lambda$. Since
$h(x)(x^n-1)$ equals $$-h_0(x)+(h_0(x)-h_1(x))x^n+\cdots +(h_{s-1}(x)-h_s(x))x^{sn}+h_s(x)x^{(s+1)n},$$
it is clear that among the powers of $x$ of degree $\geq n$, at least
 $$(|X_0|-|X_1|)+\cdots +
(|X_{s-1}|-|X_s|)+|X_s|=|X_0|=k$$ have non-zero coefficients.
\end{proof}

\begin{theorem} Let $C\subseteq\F_{q^r}^n$ be a Galois supplemented complete code and  $u=(x_1,\dots,x_n)\in C$ a word of odd type (that is, $x_1+\cdots
+x_n\not =0$). Then $\wt (u)\geq \root r\of n$.
\end{theorem}

\begin{proof}
Let $u\in C=C_g$ be a word of odd type, that we identify with a polynomial $f$ of degree less than  $n$. Since $f$ is a multiple of $g$,
$\prod_{\sigma\in\Gal(\F_{q^r}/\F_q)}f^\sigma$ is a multiple of $\prod_{\sigma\in\Gal(\F_{q^r}/\F_q)}g^\sigma=1+x+\cdots +x^{n-1}$, say
\begin{equation}\label{norma}
\prod_{\sigma\in\Gal(\F_{q^r}/\F_q)}f^\sigma (x)=c(x)(1+x+\cdots +x^{n-1}),\end{equation}
where $c(1)\not=0$, given that $f(1)\not=0$ as $f$ has odd type. Modulo $x^n-1$, the right hand side of (\ref{norma}) is $c(1)(1+x+\cdots +x^{n-1})$, so for
some polynomial $h$ we have
\begin{equation}\label{pol}\prod_{\sigma\in\Gal(\F_{q^r}/\F_q)}f^\sigma (x)=c(1)(1+x+\cdots +x^{n-1})+h(x)(x^n-1).\end{equation}
Suppose that among the first $n$ coefficients of $h$, $k$ of them are equal to $-c(1)$.
Then among the first $n$ coefficients of the right-hand side of (\ref{pol}), $n-k$ are non-zero. On the other hand Lemma \ref{ceros} guarantees that at least
$k$ coefficients of (\ref{pol}) of powers of $x$ of degree $\geq n$ are also non-zero, so a total of  at least $n$  coefficients are non-zero. As for the
left-hand side of  (\ref{pol}), notice that all the polynomials  $f^\sigma$ have the same number of non-zero coefficients, namely, $t=\wt(u)$. Hence the product
of all of them has at most   $t^r$ non-zero coefficients. The result follows now directly.
\end{proof}

\subsection{Idempotent polynomials}

We recall that the \emph{idempotent generator polynomial} of a cyclic code $C=C_g\subseteq \F_{q^r}[x]/(x^n-1)$  is the unique polynomial $e$ of degree less
than $n$ such that its image in $\F_{q^r}[x]/(x^n-1)$ is the idempotent generating $C$ as an ideal. Of course it is the polynomial of degree less than $n$ that
is  $0$ on the roots of $g$ and $1$ on the $n$th roots of unity that are not roots of $g$. Therefore
\begin{eqnarray*}e(x)&=&\frac{1}{n}\sum_{\stackrel{\scriptstyle\zeta^n=1}{\scriptstyle g(\zeta)\not=0}}\zeta\frac{x^n-1}{x-\zeta}=
\frac{1}{n}\sum_{\stackrel{\scriptstyle\zeta^n=1}{\scriptstyle g(\zeta)\not=0}}\frac{(x/\zeta)^n-1}{(x/\zeta)-1}
= \frac{1}{n}\sum_{\stackrel{\scriptstyle\zeta^n=1}{\scriptstyle g(\zeta)\not=0}} (1+\frac{x}{\zeta}+\cdots +
(\frac{x}{\zeta})^{n-1})=\frac{1}{n}\sum_{j=0}^{n-1}(\kern -4pt\sum_{\stackrel{\scriptstyle\zeta^n=1}{\scriptstyle g(\zeta)\not=0}} \zeta^{-j}) x^j\\
&=&\frac{\dim C}{n}-
\frac{1}{n}\sum_{j=1}^{n-1}(\kern -4pt\sum_{g(\zeta)=0} \zeta^{-j}) x^j,\end{eqnarray*}
where, the last equality holds because $\sum_{\zeta^n=1}\zeta^{-j}=0$ for $1\leq j\leq n-1$. If, for $X\subseteq \{1,\dots ,n-1\}$ and $\lambda\in
\F_{q^r}^\times$, we set
$$S_X(x)=\sum_{i\in X}x^i,\ \ B_\lambda=\{1\leq j\leq n-1\mid \sum_{g(\zeta)=0} \zeta^{-j}=\lambda\},$$
we can write
\begin{equation}\label{idempotentegeneral}e(x)=\frac{\dim C}{n}-\frac{1}{n}\sum_{\lambda\in\F_{q^r}^\times}\lambda S_{B_\lambda}(x).\end{equation}
The sets $B_\lambda$, $\lambda\in\F_{q^r}^\times$, partition $\{1,\dots,n-1\}$ and, by the natural identification of this set with $\Z/n\Z\backslash\{0\}$, this
is a partition into blocks by the action of the group $G=\langle q\rangle\leq (\Z/n\Z)^\times$ (given that  $  q^kB_\lambda=B_{\lambda^{q^k}}$). Moreover the
sets $B_\lambda$ are invariant by $H=\langle q^r\rangle$, and the stabilizer of $B_\lambda$ is precisely $H$ if $\lambda$ is a primitive element of the
extension $\F_{q^r}/\F_q$. In particular for a fixed primitive element $\lambda$, any Galois supplemented code $C$ (respectively, Galois coprime polynomial $g$
or block $B$) originates another Galois supplemented code $C_\lambda$ (respectively, Galois coprime polynomial $g_\lambda$ or block $B_\lambda$).  For the
quaternary extension $\F_4/\F_2$ this correspondence defines a curious duality whose statement and proof we give in terms of blocks.

\begin{theorem} Let $\lambda\in\F_4\backslash\F_2$, $\zeta$ a primitive $n$th root of unity and  $B\subseteq\Z/n\Z$ a complete $2^2$-block. We put
$$B^*=B_\lambda=\{1\leq j\leq n-1\mid \sum_{i\in B}{(\zeta^i})^{-j}=\lambda\}.$$ Then $B^*\subseteq\Z/n\Z$ is a complete $2^2$-block too and   $B^{**}=-B$ or
$-2B$ according to $n\equiv 1$ or $3 \pmod 4$, respectively. In particular the complete Galois supplemented codes over $\F_4$ of length $n$ are exactly the
cyclic codes with idempotent generator polynomials of the type
\begin{equation}\label{generalform}\frac{n+1}{2}+\lambda\sum_{i\in B}x^i+\lambda^2\sum_{i\in B'}x^i,\end{equation}
where $B\subseteq\{1,\dots, n\}$ is a complete  $2^2$-block and $B'$ is its complementary set in $\{1,\dots ,n-1\}$.
\end{theorem}

\begin{proof}
We already know  that $B^*$ is a block. We show now that it is complete. Let $1\leq j\leq n-1$. Then
$$\sum_{i\in B}\zeta^{-ij}+(\sum_{i\in B}\zeta^{-ij})^2=\sum_{i\in B}\zeta^{-ij}+\sum_{i\in 2B}\zeta^{-ij}=\sum_{i=1}^{n-1}(\zeta^{-j})^i=1,$$
so $\sum_{i\in B}\zeta^{-ij}=\lambda$ or $\lambda^2$. Hence $j\in B^*\cup 2B^*$, that is $B^*$ is complete and, by (\ref{idempotentegeneral}),
\begin{equation}\label{idem}e(x)=\frac{n+1}{2}+\lambda\sum_{j\in B^*}x^j+\lambda^2\sum_{j\in 2B^*}x^j\end{equation}
is the idempotent polynomial of $C_B$.

Next we show that $B^{**}=-B$ or $-2B$. As $B^*$ is a complete block we can apply the result we have just proved to conclude the completeness  of    $B^{**}$,
which means that for all $1\leq i\leq n-1$, $\sum_{j\in B^*} (\zeta^j)^{-i}=\lambda$ or $\lambda^2$,  or, writing $i$ for $-i$, $\sum_{j\in B^*}
(\zeta^j)^i=\lambda$ or $\lambda^2$. If $i\in B$, $\zeta^i$ is a root of the idempotent (\ref{idem}), so evaluating in $x=\zeta^i$ we get
$$\lambda\sum_{j\in B^*}(\zeta^i)^j+\Big(\lambda\sum_{j\in B^*}(\zeta^i)^j\Big)^2=\frac{n+1}{2}=1 \text{ or } 0 ,$$
according to $n\equiv 1$ or $3\pmod 4$, respectively.
We conclude that $\sum_{j\in B^*}(\zeta^i)^j$ equals $\lambda$ in the former case and    $\lambda^2$ in the latter, that is   $-i\in B^{**}$ or $-2i\in B^{**}$.
So $-B\subseteq B^{**}$ or $-2B\subseteq B^{**}$. But both sets $B$ and $B^{**}$ have the same number of elements, being both of them complete, so these
inclusions are in fact equalities. This  proves of course that (\ref{generalform}) is the idempotent polynomial of the Galois supplemented code   $C_{-B^*}$ or
$C_{-2B^*}$.
\end{proof}

The last theorem shows that  Galois supplemented complete codes over $\F_4$ are $Q$-codes in the sense of \cite{Pless}.

\subsection{Automorphisms}

Let $C$ be a cyclic code of length $n$ with generator polynomial $g$ and roots  $\zeta^i$, $i\in B\subseteq\Z/n\Z$ ($\zeta$ is a primitive $n$th root of unity).
If $K$ is the setwise stabilizer of $B$ under the action of  $(\Z/n\Z)^\times$, any affine transformation $x\mapsto \pi_{a,b}(x)=ax+b$,   $a\in K$,
$b\in\Z/n\Z$,  defines naturally an automorphism of the code $C$. Indeed if $(c_0,\dots ,c_{n-1})\in C$, $i\in B$ and $a'\in K$ is the inverse of $a$ in
$(\Z/n\Z)^\times$,
$$\sum_{j=0}^{n-1} c_{aj}(\zeta^i)^j=
\sum_{j=0}^{n-1}c_j(\zeta^i)^{a'j}=
\sum_{j=0}^{n-1}c_j(\zeta^{a'i})^j=0.$$
In particular, as $\pi_{-1,1}$ corresponds to the reversion map $(c_0,c_1,\dots ,c_{n-1})\mapsto (c_{n-1},\dots ,c_1,c_0)$, it follows that $C$ is reversible
(i.~e., invariant under this map) if and only if $-B=B$. Reversibility is an important property for the applications to DNA-codes that we'll explain later.

\begin{theorem} Let $C=C_B$ be a Galois supplemented code of length $n$ over $\F_{q^r}$ and $K$ the stabilizer of $B$ in $(\Z/n\Z)^\times$. Then the group of
affine transformations $\pi_{a,b}$ with $a\in K$, $b\in \Z/n\Z$, is contained in the group of automorphisms of $C$. In particular $C$  is reversible if and only
if $-B=B$. If $n$ is an odd prime power, then  $C$ is reversible if the order of $q$ modulo $n$ is a multiple of $2r$ and this condition is also necessary if
the order of $q$ modulo $n$ is even.
\end{theorem}

\begin{proof}
Only the second part has to be proved. If $n$ is an odd prime power the group $(\Z/n\Z)^\times$ is cyclic, so the three conditions
 $2r\mid\ord (q\mod n)$; $\ord(q^r\mod n)$ is even and   $-1$ lies in $H$, the subgroup generated by $q^r$, are equivalent. So, if this is the case, as $B$ is
 invariant by $H$, $B$ is fixed by $-1$, that is, $C$ is reversible. If $\ord(q\mod n)$ is even, $-1\in G$, the subgroup generated by $q$, so if $C$ is
 reversible, $-1\in H$ and $\ord (q \mod n)$ is a multiple of $2r$.
\end{proof}

In general, for fixed $q$, $r$ and $n$ with $r\mid \ord (q \mod n)$, there exist Galois supplemented reversible codes if and only if $-1\not\in G\backslash H$.
Necessity is clear. Sufficiency is also obvious if $-1\in H$. Finally, if    $-1\not\in G$,   we consider a set $T$ of representatives of the action of $G$ on
$\Z/n\Z\backslash\{0\}$ such that $T=-T$. We can take for $B$ the union of the $H$-orbits of the elements of any $T'\subseteq T$ such that $T'=-T'$. Notice that
$-1\not\in G\backslash H$ when the order of $q$ modulo $n$ is odd or when it is a multiple of $2r$ (because in this last case $H$ has even order, so $-1$ is in
$G$ if and only if it is in $H$).

\subsection{Duality}

Given a code $C\subseteq\F_q^n$, we keep, as in Section 3, the notation $C^\perp$ and $\widetilde C$ for the dual and extended code of $C$, respectively.

\begin{lemma}\label{dual} Suppose $\F_q$ is a field of characteristic $2$ and  $x^n-1=(x-1)g(x)h(x)$. We denote by $C$ and $C'$ the cyclic codes with generator
polynomials $g$ and $h^R$, respectively. Then $\widetilde C^\perp=\widetilde C'$ (for $\widetilde C^\perp$, extend first).
\end{lemma}

\begin{proof}
The codes $\widetilde C^\perp$ and $\widetilde C'$ have the same dimensions since
\begin{eqnarray*}
\dim \widetilde C^\perp &=& n+1-\dim\widetilde C=n+1-\dim C=1+\deg g=n-\deg h\\ &=&n-\deg h^R=\dim C'=\dim\widetilde C'.\\ \end{eqnarray*}
Under the natural identification of $\F_q^n$  and the space of polynomials with degree less than $n$, $\PP_{n-1}$, the standard inner product of two polynomials
$f_1,f_2\in \PP_{n-1}$, $(f_1,f_2)$, becomes the independent term of the Laurent polynomial $f_1(x)f_2(1/x)$. Now, for $0\leq i\leq r$, $0\leq j\leq n-r+1$,
$r=\deg h$:
$$(x^ig(x))(x^{-j}h^R(1/x))=x^{-(r-i+j)}g(x)h(x)=x^{-(r-i+j)}(1+x+\cdots +x^{n-1}),$$
$0\leq r-i+j\leq n-1$, so $(x^ig(x),x^jh^R(x))=1=g(1)h(1)=g(1)h^R(1)$ (to evaluate $g(1)h(1)$ notice that $n$ is odd and the characteristic of the field is
$2$). This means that the two extended codewords  $(x^ig(x), g(1))\in\widetilde C$ and  $(x^jh^R(x), h^R(1))\in\widetilde C'$ are orthogonal, that is
$\widetilde C^\perp\subseteq \widetilde C'$. As the two spaces have the same dimensions, equality follows.
\end{proof}

\begin{corollary} Let $q$ be a power of $2$ and $B\subseteq\Z/n\Z$ a  complete $q^2$-block. Then $\widetilde C_B^\perp=\widetilde C_{-2B}$. In particular
$\widetilde C_B$ is self-dual if and only if $B=-2B$. If $n$ is a prime power,  this condition holds if and only if the order of   $q$ modulo $n$ is congruent
to $2$ modulo $4$.\end{corollary}

\begin{proof}
By the hypotheses $x^n-1=(x-1)g_B(x)g_B^\sigma(x)$ so, by Lemma \ref{dual}, $\widetilde C_B^\perp$ is the extended code of the cyclic code with generator
polynomial the reciprocal polynomial of  $g_B^\sigma=g_{2B}$, which is   $g_{-2B}$. The first part of the corollary follows immediately. If $n$ is an (odd)
prime power, $(\Z/n\Z)^\times$ is cyclic and, as $|G:H|=2$, $-1\in G$  (we keep the usual notation for $G$ and $H$). But  $2\not\in H$, so $-2\in H$ if and only
if $-1\not\in H$. This is equivalent to $H$ having odd order, which in turn is the same as the order of $G$ being congruent to $2$ modulo $4$. So if this
condition holds it is clear that $B=-2B$. On the other hand, if it doesn't,   $-1\in H$ and certainly  $-2B=2B\not=B$.
\end{proof}

\section{Applications to DNA-codes}

Recall that $A_4^{GC,RC}(n,d)$ is the maximum number of words in a DNA-code $C$ of length $n$ with constant $GC$-content such that the minimum distance is $\geq
d$ and satisfies the $d$-reverse-complement constraint condition, $d(u,v^{RC})\geq d$ for all $u,v\in C$. Of course, if $C$ is any  code over $\F_4$ with odd
length $n$, then $v\not=v^{RC}$ for any codeword $v\in C$ (because reversion and complementation add one the central coordinate), so we can partition $C$ into
two  subsets $C_1$ and $C_2$ of equal size such that $u\not= v^{RC}$ for all $u,v\in C_1$ or $C_2$. If $C$ is reversible and complemented with minimum distance
$\geq d$,
it is clear that any of the subcodes $C_1$ or $C_2$ satisfies the $d$-reverse-complement constraint and its minimum distance is certainly $\geq d$. We can take
as $C$ a Galois supplemented reversible code over $\F_4$ that, according to Theorem \ref{principal}, has exactly $2^{n-2\deg g}{n \choose {n+1\over 2}}$ words
of $\F_2$-weight $(n+1)/2$. Thus the code $C_1$ (or $C_2$) would have exactly $2^{n-2\deg g-1}{n \choose {n+1\over 2}}$ words of $\F_2$-weight $(n+1)/2$.

Notice that in the construction above, the code  $C_1$ (or $C_2$) cannot be a linear code  over $\F_4$ (though, with some care, one can choose it to be
$\F_2$-linear, namely take any $\F_2$-hyperplane of $C$ not containing $(1,\dots,1)$). It could be desirable to find our code of constant GC-content inside a
genuine linear subcode of $C$ (but of course,  keeping it as large as the one we already have). It turns out that we can do it by considering $C^+$: since $n$
is odd, $C=C^+ \oplus \langle (1,\dots ,1)\rangle$, so for any $v\in C^+$, $v^{RC}\not\in C^+$, whence $v\not=v^{RC}$  and the $d$-reverse-complement constraint
$d(u,v^{RC})\geq d$ is valid for all $u,v\in C^+$.  By Corollary \ref{enumeratorplus}, $C^+$ contains $2^{n-2\deg g-1}{n \choose {n+1\over 2}}$ words of
$\F_2$-weight $(n+1)/2$ or $(n-1)/2$ (whichever of the two that is odd; notice that Corollary \ref{enumeratorplus} says that for odd $\F_2$-weight, half of the
words of $C$ are in $C^+$ and the other half in $C\backslash C^+$, but for even $\F_2$-weight, all of them are in $C\backslash C^+$), so we finally get a code
of the same size as before.

Any of the two strategies lead then to the following result.

\begin{theorem} Suppose that there exists a Galois supplemented reversible code $C\subseteq \F_4^n$ of minimum distance $d$. Then
$$A_4^{GC,RC}(n,d)\geq  {n \choose {n+1\over 2}}.$$
\end{theorem}

Now, using for instance the package GUAVA of the computer system GAP \cite{Guava} or Magma \cite{Magma}, we can improve several of the bounds for
$A_4^{GC,RC}(n,d)$ given in \cite{Aboluion}. For $n=17$ we can take the $2^2$-block $B=\{2,8,9,15\}$. The corresponding Galois supplemented code has minimum
distance $4$, so the constructions explained in the two previous paragraphs show that $A_4^{GC,RC}(17,4)\geq 2^8{17 \choose 8}=6223360$ (the bound given in
\cite{Aboluion} is $3111120$). Similarly, if we take $B=\{2,6,7,8,9,10,11,15\}$,  we get that $A_4^{GC,RC}(17,7)\geq {17 \choose 8}=24310$ (\cite{Aboluion}
gives $12120$). Notice that the Galois supplemented code $C_B$, while having the same dimension and distribution of $\F_2$-weights as the corresponding QR-code,
has bigger minimum distance ($7$ and $5$, respectively).
Also, considering the QR-codes for $n=13$ and $29$, we obtain the improved bounds $A_4^{GC,RC}(13,5)\geq 1716$ and $A_4^{GC,RC}(29,11)\geq 77558760$.





\bibliographystyle{elsarticle-num}
\bibliography{<your-bib-database>}



\end{document}